%% file: root.tex
\documentclass[letterpaper,10pt,conference]{IEEEconf}
\usepackage{amsmath,amsthm}
\usepackage{hyperref}
\usepackage{graphicx}
\usepackage{xargs}
\usepackage{latex-macros}
\usepackage[colorinlistoftodos,prependcaption]{todonotes}
\usepackage[labelfont=bf]{caption}
\usepackage{listings}
\usepackage{color}
\usepackage{cite}

\usepackage{algorithm}
\usepackage{algpseudocode}

\IEEEoverridecommandlockouts                              % This command is only
\definecolor{codegreen}{rgb}{0,0.6,0}
\definecolor{codegray}{rgb}{0.5,0.5,0.5}
\definecolor{codepurple}{rgb}{0.58,0,0.82}
\definecolor{backcolour}{rgb}{0.95,0.95,0.92}

 \hypersetup{
 	colorlinks=true, %set true if you want colored links
 	linktoc=all,     %set to all if you want both sections and subsections linked
 	linkcolor=blue,  %choose some color if you want links to stand out
 }

\lstdefinestyle{mystyle}{
    backgroundcolor=\color{backcolour},   
    commentstyle=\color{codegreen},
    keywordstyle=\color{magenta},
    numberstyle=\tiny\color{codegray},
    stringstyle=\color{codepurple},
    basicstyle=\footnotesize,
    breakatwhitespace=false,         
    breaklines=true,                 
    captionpos=b,                    
    keepspaces=true,                 
    numbers=left,                    
    numbersep=5pt,                  
    showspaces=false,                
    showstringspaces=false,
    showtabs=false,                  
    tabsize=2
}
\newcommandx{\question}[2][1=]{\todo[inline,backgroundcolor=red!25,bordercolor=red,#1]{#2}}

\begin{document}
 
\title{\LARGE \bf{Reducing Aggregate Electric Vehicle Battery Capacity through Sharing}}
\author{Polina Alexeenko$^{\dagger}$ and Vasileios Charisopoulos$^{*}$%
\thanks{ $^{\dagger}$Electrical \& Computer Engineering, Cornell University, Ithaca NY, \url{pa357@cornell.edu}}
\thanks{$^{*}$ Operations Research \& Information Engineering, Cornell University, Ithaca NY, \url{vc333@cornell.edu}}
}

\maketitle

\begin{abstract}
   Meeting growing demand for automotive battery resources is predicted to be costly from both economic and environmental perspectives. 
   To minimize these costs, battery resources should be deployed as efficiently as possible.
   A potential source of inefficiency in battery deployment is the fact that the batteries of personal vehicles are typically much larger than needed to meet most daily mobility needs.
   In this paper, we consider whether battery resources can be used more efficiently in a setting where drivers, in addition to having personal vehicle batteries, have access to a shared battery resource. 
   More precisely, we consider the problem of minimizing aggregate battery capacity in settings with and without a shared resource subject to the requirement that driver commuting needs are met with high reliability.
   To assess the potential for reductions in deployed battery capacity with the addition of a shared resource, we quantify the difference in deployed battery capacity with and without a shared resource in case study using real-world longitudinal mobility data from Puget Sound, Washington.
   We find that giving drivers access to a shared battery resource can substantially reduces deployed battery capacity. 
   Furthermore, relative reductions in battery capacity increase with number of drivers and the level of reliability desired.
\end{abstract}

\newtheorem{theorem}{Theorem}[section]
\newtheorem{proposition}{Proposition}[section]
\newtheorem{lemma}{Lemma}[theorem]
\newtheorem{definition}{Definition}[section]
\newtheorem{example}{Example}[section]
\newtheorem{remark}{Remark}[section]

\input{outline.tex}

%% REFERENCES
\bibliographystyle{IEEEtran}
\bibliography{References}

\appendix
\input{appendix}

\end{document}

%% file: outline.tex
\section{Introduction}

\input{introduction}

\section{Formulation} \label{sec:formulation}

\input{formulation.tex}

\section{Results} \label{sec:results}

\input{results}

\section{Empirical study} \label{sec:study}

\subsection{Data source and model}

\input{data_model.tex}

\subsection{Study results}

\input{study_results.tex}

\section{Conclusion} \label{sec:conclusion}

\input{conclusion.tex}

%% file: introduction.tex
Meeting the growing demand for electric vehicle (EV) batteries is
predicted to be costly from both economic and environmental
perspectives~\cite{notter2010contribution,olivetti2017lithium}. 
In 2015, global automotive battery production was less
than $40$ GWh/year. By 2020, battery
production had increased four-fold to over 150 GWh/year~\cite{outlook2021electric}.
Increases in battery production are driven by multiple factors.
First, the number of electric vehicles (EVs) on the road is growing due to a combination of government incentives, increasing environmental concerns, and decreasing vehicle costs. 
Secondly, vehicle battery sizes are themselves becoming larger: the average range of a new BEV increased by 43\% between 2015 and 2020, from 124 miles to 218 miles~\cite{outlook2021electric}.

Demand for large battery capacities stems, in part, from \textit{range anxiety}, the fear that an EV's battery capacity will be insufficient to satisfy a typical driver's mobility needs. 
That is, because of the sparsity of EV charging stations and the long charging times associated with most stations, drivers tend to chose their battery sizes to minimize the likelihood of exhausting their battery mid-trip. 
As a result, typical commuting distances of EV owners tend to be small relative to battery size~\cite{zhao2021assessment,AB23}. 
For example, a study of the driving behavior of several hundred drivers in the United States over the course of a year found that 75\% of drivers traveled \emph{fewer than 100 miles daily} during 96\% of the year~\cite{pearre2011electric}.

Instead of sizing batteries to meet drivers' infrequent long-distance travel needs, EV drivers can use \textit{range extenders}.
Range extenders (REs) are auxiliary devices which provide additional energy to the EV to supplement its battery and increase range \cite{tran2021review}.
REs can be used to reduce vehicle battery sizes while ensuring that both long and short distance commuting needs are met.  
It is important to note that although REs can offset personal vehicle battery sizes, their impact on the system as a whole (e.g., in terms of environmental or economic benefits) can be difficult to measure because of the diversity of RE technologies available.
For example, range extender powertrains vary widely, including internal combustion engines, hydrogen fuel cells, and gas turbines \cite{guanetti2016energy, tran2021review,wu2019intelligent}. 
Of particular relevance to this paper are range extenders consisting of a trailer-mounted battery which can be plugged into the EV through its charging port \cite{segard2016ep}.
Because these REs use the same powertrain as the vehicles themselves, their capacity can be compared directly with personal vehicle battery capacity.

When personal vehicle batteries are sized to meet typical commuting needs and when driver commuting distances are not strongly correlated, the probability that many drivers will simultaneously require range extension is low.
This suggests that a modestly sized RE resource could be shared across drivers without compromising reliability.
Furthermore, providing drivers with a shared battery resource (in addition to personal batteries) could reduce the total (i.e., personal and shared RE) amount of battery capacity deployed relative to a setting where drivers only have personal batteries.
The focus of this paper is to quantify the potential reduction in total battery capacity achievable through the introduction of a shared RE resource. 
In particular, we consider the problem of determining the amount of battery capacity required to meet driver commuting needs with high probability in settings with and without a shared resource and show that the presence of a shared resource can reduce total battery capacity without reducing reliability.

The concept of using a shared resource to reduce system-wide risk is related to the concept of \textit{diversification} in financial risk management. 
Diversification refers to the phenomenon where financial assets with varying risk profiles are combined into a portfolio whose aggregate level of risk is lower than the sum of its component assets \cite{mcneil2015quantitative, mathur2015risk}.
In the context of power systems, risk aggregation (with the objective of improving system robustness) has received considerable attention, e.g., in the context of wind generation  \cite{qiu2013risk, baeyens2013coalitional} , photovoltaics \cite{alafita2014securitization, lowder2013potential}, and mini-grids \cite{united2015increasing, malhotra2017scaling}.
Closest in spirit to our work is that of Abdolmaleki et al., which considers increasing vehicle range using a network of vehicles able to share power through wireless transfer \cite{abdolmaleki2019vehicle}. 
Although the authors mention the potential of their proposed methodology to reduce battery capacity, the discussion is brief and focuses on quantifying reliability under a specific alternative battery capacity size. 
This differs from the more general framework presented in our paper, which aims to characterize the capacity-reliability tradeoff for a wider range of capacity sizes.

\subsection{Contributions}

In this paper, we explore how the presence of a shared battery resource can reduce the total amount of battery capacity deployed across a system of EV drivers while guaranteeing that every driver's mobility needs are met with high reliability. 
We formulate the battery capacity planning problem as a chance-constrained optimization problem and derive a conservative approximation of this problem.
Our approximation offers two key advantages over the original problem.
First, the constraint function is convex in the decision variables and thus amenable to solution using scenario approximation.
Second, while the original problem involves a number of constraints equal to the number of drivers in the system, our reformulation has only a single constraint.
To demonstrate the practical utility of the proposed framework, we assess the potential for capacity reduction through sharing using real-world mobility data from Puget Sound, Washington. 
Our empirical results suggest that access to a shared resource can significantly reduce
the required battery capacity, and that the potential for reduction increases with desired reliability levels and the number of drivers in the system.

\subsection{Notation and organization}
We briefly describe notation used throughout the paper. Vectors are denoted
using lowercase letters in boldface, such as $\bm{x}$.
We write $(x)_+ := \max\set{x, 0}$ for the positive part
of a scalar. We write $\Pr(\cE)$  for the probability of an event $\cE$, and $\expec{X}$ for the expectation of random variable $X$. 
Finally, given a positive integer $n$, we
write $[N]$ for the set $\set{1, 2, \dots, N}$.

The remainder of the paper is organized as follows. Section~\ref{sec:formulation} introduces the system model and presents the optimization problems in the shared and non-shared settings. 
Section~\ref{sec:results} presents the conservative reformulation of the battery capacity planning problem and discusses approximate solutions to the reformulated problem.
Section~\ref{sec:study} presents a results from an empirical study on real-world commuting data, and Section~\ref{sec:conclusion} concludes and discusses
potential future directions.

%% file: formulation.tex
\subsection{System model}

We consider a setting in which a central decision maker selects a set of electric vehicle battery capacities to serve the mobility needs of a group of EV drivers. 
The decision maker's objective is to minimize the total amount of EV battery capacity deployed (in order to minimize the economic or environmental costs of battery capacity production) while ensuring that each driver's daily energy requirements are satisfied with high probability. 
We index the set of driver's according to $i \in [\numCust] := \left\{1, \dots, \numCust \right\}$ and denote their personal battery capacity by $\batSizPriv_i$.
Each driver's daily energy requirements are distributed according to a probability distribution $\Pbb$, and they require that the amount of battery capacity available to them exceed their daily energy requirement with probability at least $\probSucc \in \paren{0,1}$.

In the setting without a shared resource, each driver's personal battery capacity must exceed their daily energy requirement with high probability, and the decision maker selects capacities according to the solution of the following optimization problem:
\begin{equation}
    \tag{P-NS}
    \begin{aligned}
    \underset{\bm{\batSizPriv}}{\text{minimize}} \quad & \sum_{i=1}^{\numCust} \batSizPriv_i \\
    \text{subject to} \quad & \Pr \left(  \tripDist_i < \batSizPriv_i \right)   \geq \probSucc_i \quad \forall i \in [\numCust].
    \end{aligned}
    \label{eq:opt-non-shared}
\end{equation}
The optimal solution to Problem \eqref{eq:opt-non-shared} sets each driver's battery capacity according to the quantile of their daily energy distribution, i.e.,
\begin{equation}
      \batSizPriv_i = F^{-1}_i (\alpha_i) := \inf \{ x \in \Rbb : F_i\paren{x} \geq \alpha_i \}.
\end{equation}
where $F_i$ denotes the cumulative distribution function associated with driver $i$'s energy requirement distribution $\Pbb_i$. 
In aggregate, the total amount of battery capacity required to meet everyone's driving needs is given by 
\begin{align}
      \mathsf{opt}_{\mathrm{ns}}^* = \sum_{i=1}^\numCust F^{-1}_i (\alpha_{i}).
      \label{eq:ns-quants}
\end{align}

When the target reliability $\probSucc_i$ is high, driver batteries tend to be under-utilized in the sense that battery capacities are much larger than required to meet mobility needs on most days.
Furthermore, when energy requirements are not strongly correlated across drivers, the probability that a large fraction of drivers have high energy requirements on the same day is small. 
Together, these observations suggest that under certain circumstances, diverting resources from personal batteries to a shared battery may allow for significant reductions in the total battery capacity needed to satisfy driver reliability constraints.

\subsection{Battery sharing as a chance constrained problem}

As an alternative to the setting where drivers must rely exclusively on their personal batteries, we consider a setting in which drivers have access to a shared battery resource of capacity $\batSizPub$. 
We assume that this resource can be divided into units of arbitrary size and distributed to the individual drivers at no cost. 
In this setting, the decision maker must chose both the personal battery sizes $\batSizPriv_i \text{ for } i \in \mathcal{N}$ and the size of the shared resource $\batSizPub$. 
In the presence of a shared resource, the probability with which a driver's needs are satisfied is a function of both their personal battery size and the portion of shared resource allocated to them.
In this setting, the decision maker chooses shared and personal battery sizes as the solution to the following:
\begin{equation}
    \tag{P-S}
        \begin{aligned}
        \underset{\batSizPub, \bm{\batSizPriv}}{\text{minimize}} \quad
        &\batSizPub + \sum_{i=1}^{\numCust} \batSizPriv_i \\
        \text{subject to} \quad & \Pr \left(\batSizPriv_i + f_i(\bm{\batSizPriv}, \bm{\tripDist}, \batSizPub) \geq \tripDist_i \right) \geq \alpha_i \quad \forall i \in [\numCust]
    \end{aligned}
    \label{eq:opt-shared}
\end{equation}
where $f_i: \Rbb^{N} \times \Rbb^N \times \Rbb \mapsto \Rbb$ is an \emph{allocation rule} determining the quantity of shared resource allocated to driver $i$ as a function of the personal battery capacities, shared resource capacity, and realized energy requirements of all drivers.
At a slight abuse of terminology, we will refer to both $f_{i}$ and the collection
$\bm{f} \equiv (f_i)_{i=1}^N$ as allocation rules in the sequel.

Problem \eqref{eq:opt-shared} belongs to a class of problems known as \textit{chance constrained} programs. 
Chance constrained problems can be non-convex even in simple settings (e.g., when the constraint functions $g_i$ are affine in the decision variables $\batSizPub, \batSizPriv$ and uncertain parameters $\tripDist$).
In settings where little is known about the underlying distribution or the distribution is intractable to analyse directly, \textit{scenario-based} approximations to chance constrained problems are particularly useful \cite{campi2008exact,calafiore2010random}. 
Scenario approximations involve replacing chance constraints with a set of sampled constraints.
In order to produce a solution feasible for the original problem with high confidence, scenario approximations require only that the constraint function is convex in the decision variable.
In our setting, the convexity of the constraint function depends on the rule or policy according to which the shared resource is allocated. 
In the sequel, we discuss common allocation rules of practical interest and their implications on the convexity of the constraint functions.

\subsection{Shared resource allocation model} 
\label{sec:shared-resource-allocation-model}
The amount of resource available to a driver is the sum of their personal battery and the quantity of shared resource allocated to them. 
We assume that the allocation rule is structured to ensure that the sum over each driver's shared resource allocation does not exceed the total quantity of shared resource, i.e,
$
\sum_{i=1}^N f_{i}(\bm{\batSizPriv}, \bm{\tripDist}, \batSizPub) \leq \batSizPub.
$
Furthermore, we assume that allocations must be non-negative, i.e., $f_{i}(\cdot) \geq 0$.
Beyond these mild assumptions, the structure of the allocation rule is highly flexible.

However, even simple allocation rules (e.g., proportional allocation) can result in intractable reliability constraints, rendering Problem~\eqref{eq:opt-shared} unamenable to scenario-based approximations.\footnote{A discussion of the intractability of the reliability constraints resulting from proportional allocation rules appears
in Appendix~\ref{sec:app-nonconvex}.}
In order to overcome the challenges presented by many allocation rules, we propose an approximation to Problem
\eqref{eq:opt-shared} and show that our proposal furnishes an \emph{inner approximation}
to the constraint set of~\eqref{eq:opt-shared} for a large family of allocation functions.

%% file: results.tex
\subsection{Tractable approximation} 

Instead of attempting to solve Problem~\eqref{eq:opt-shared} directly, we approximate it by a simpler, though often more conservative problem.
Specifically, we replace the $\numCust$ individual constraints on energy requirement satisfaction with a single constraint requiring that the \textit{aggregate shortfall} (the sum over each driver's shortfall) be less than the shared battery capacity with high probability.
That is, we approximate Problem~\eqref{eq:opt-shared} by
\begin{equation}
    \tag{P-SI}
    \begin{aligned}
    \underset{\bm{\batSizPriv}, \batSizPub}{\text{minimize}} \quad & \batSizPub + \sum_{i = 1}^N \batSizPriv_i \\
    \text{subject to}               \quad & \Pr\left(
        \sum_{i = 1}^N (\tripDist_i - \batSizPriv_i)_{+} \leq \batSizPub
    \right) \geq \alpha
    \end{aligned}
    \label{eq:inner}
\end{equation}
where $\alpha = \max_i (\probSucc_1, \dots, \probSucc_{\numCust})$. 

Problem~\eqref{eq:inner} offers two key structural advantages over the original capacity planning problem~\eqref{eq:opt-shared}. 
First, while the original problem involves $N$ constraints, Problem~\eqref{eq:inner} uses a single constraint. 
Second, the function inside the chance constraint is convex, making Problem~\eqref{eq:inner} more amenable to both direct analysis and approximation.

Furthermore, Problem~\eqref{eq:inner} is an inner approximation to the original capacity planning problem for many common allocation rules.
That is, feasible solutions to Problem~\eqref{eq:inner} are feasible for the original capacity planning Problem~\eqref{eq:opt-shared} for a large family of allocation rules which we refer to as \textit{shortfall-minimizing} allocation rules.

\begin{definition}[Shortfall-minimizing allocation rules]
    Consider an allocation rule $f = (f_i)_{i=1}^N$, where
    $f_i: \Rbb^N \times \Rbb^N \times \Rbb \to \mathbb{R}_+$. We call $f$
    \emph{shortfall-minimizing} if it satisfies the following property: 
    \begin{equation}
      \batSizPub \geq \sum_{i=1}^N (\tripDist_i - \batSizPriv_i)_+ \implies
    f_{i}(\bm{\batSizPriv}, \bm{\tripDist}, \batSizPub) \geq \tripDist_i -   \batSizPriv_{i}, \;\; \text{for all $i \in [N]$},
      \label{eq:allocation-rule-property}
    \end{equation}
\end{definition}
The above definition simply states that whenever the amount of shared resource exceeds the aggregate shortfall across drivers (i.e., when there is enough shared resource to satisfy all driver shortfalls), each driver receives a portion of the shared resource greater than or equal to their energy shortfall.
The family of shortfall-minimizing rules is large, including the proportional, first-come-first-serve (FCFS), and utilitarian rules described in detail in Section~\ref{sec:study}.%
\footnote{For an example of a rule which is \textit{not} shortfall-minimizing, one can consider an `equal allocation' rule which disburses the same amount of the shared resource to all drivers regardless of their shortfall.}

\begin{theorem} For any shortfall-minimizing allocation rule, any feasible solution to Problem~\eqref{eq:inner} is feasible for Problem~\eqref{eq:opt-shared}.
\label{theorem:inner-approx}
\end{theorem}

\begin{proof}
We will write $\mathcal{E} := \set{\sum_{i} (\tripDist_i - \batSizPriv_i)_+ \leq \batSizPub}$ and
$\cE_{i} := \set{\batSizPriv_i + f_{i}(\bm{\batSizPriv}, \bm{\tripDist}, \batSizPub) \geq \tripDist_i}$ for
brevity. We have
\begin{align}
  \prob{\cE_i}
  &= \prob{\cE_i \cap \cE} + \prob{\cE_i \cap \cE^c} \notag \\
  & \geq
  \prob{\cE_i \cap \cE}
  \notag \\
  &=
  \prob{\batSizPriv_i + f_{i}(\bm{\batSizPriv}, \bm{\tripDist}, \batSizPub) \geq \tripDist_i \mid \cE}
  \cdot \prob{\cE} \label{eq:inner-approx-lb-simpl} \\
  &= \prob{\cE},
  \label{eq:inner-approx-lb-1}
\end{align}
using~\eqref{eq:allocation-rule-property} in~\eqref{eq:inner-approx-lb-simpl} to
deduce that the conditional probability $\prob{\cE_{i} \mid \cE} = 1$. Since $\probSucc \geq \probSucc_i$ and the
choice of index $i$ was arbitrary, the desired claim follows.
\end{proof}

Although the approximation can be conservative relative to the original capacity planning problem, it provides a useful tool for assessing the potential benefits of resource sharing.
That is, the aggregate battery capacity required to satisfy the constraints of Problem~\eqref{eq:inner} with a particular target reliability are often smaller than those required to achieve the same reliability level in the non-shared setting. 
Moreover, because Problem~\eqref{eq:inner} is an inner approximation of Problem~\eqref{eq:opt-shared}, analysis of Problem~\eqref{eq:inner} allows us to derive lower bounds on the reduction in deployed battery capacity achievable through resource sharing. 
For example, in Appendix~\ref{sec:app-gaus_ex}, we characterize the benefits of resource sharing in a setting where daily energy requirements are distributed according to independent Gaussians.

In settings where direct analysis of the chance constraint is not possible, a solution to Problem~\eqref{eq:inner} can be provably approximated through scenario-based methods.
Specifically, to produce an approximate solution to Problem~\eqref{eq:inner}, we replace the chance constraint with a set of $M_{\mathrm{sc}}$ \textit{sampled constraints} to produce an approximated problem:
\begin{equation}
    \begin{aligned}
        \underset{\bm{\batSizPriv}, \batSizPub}{\text{minimize}} \quad &
        \batSizPub + \sum_{i=1}^N \batSizPriv_i \\
        \text{subject to} \quad &
        \sum_{i=1}^N \max(\tripDist_i^{(j)} - \batSizPriv_i, 0) \leq \batSizPub,
        \;\; \text{$j \in [M_{\mathrm{sc}}]$.}
    \end{aligned}
    \label{eq:inner-scenario}
\end{equation}
where $\bm{\tripDist}^{(j)} \sim \Pbb$ for $j = 1, \dots, M_{\mathrm{sc}}$ are independent samples of $\bm{\tripDist}$ drawn from the underlying distribution.

A solution to Problem~\eqref{eq:inner-scenario} is guaranteed to be feasible for Problem~\eqref{eq:inner} with high confidence given a sufficiently large sample size. 
Specifically, to produce a solution which is feasible with confidence $1 - \delta$, it is sufficient to choose a sample size of at least $O\paren{\frac{N}{1-\alpha} \ln \paren{\frac{1}{\delta}}}$~\cite{calafiore2006scenario}.

\subsection{A heuristic for reducing conservatism}
In practice, the solutions obtained through scenario approximations can be very conservative.
Indeed, due to the high number of samples used in the approximation, Problem~\eqref{eq:inner-scenario} will often produce battery configurations attaining a reliability level significantly greater than the target $\alpha$. 
To reduce the conservatism of our approximations, we use a heuristic method for reducing the number of constraints involved in the solution of the problem.

 The conservatism reduction heuristic is implemented as follows. 
 We start by solving Problem~\eqref{eq:inner-scenario} using the sample size dictated by \cite{calafiore2006scenario}. 
We evaluate the empirical reliability of the obtained candidate solution using an additional set of samples of size $M_{\mathrm{eval}}$, where the sample size requirement is chosen as described in Appendix Section \ref{sec:app-rel}.  
 If the empirical reliability level is close to the target, the algorithm terminates. 
 If, however, the empirical reliability is larger than the target level, we reduce the number of samples used in the solution of Problem~\eqref{eq:inner-scenario} and re-solve the problem to produce a less conservative solution. 
 The resultant solution's empirical reliability is then evaluated, and the constraint number is either decreased or increased depending on whether the empirical reliability is greater or less than the target. 
 The conservatism reduction heuristic thus performs a binary search over the number of samples used in the scenario approximation to produce a solution with reliability level close to the target $\alpha$.
 The pseudocode for the conservatism reduction heuristic is given in Algorithm~\ref{alg:conservatism-reduction}.

\begin{algorithm}[h]
    \caption{Conservatism reduction heuristic}
    \begin{algorithmic}[1]
        \State \textbf{Inputs}:
        \begin{tabular}[t]{l l}
             $M_{\mathrm{sc}}$ & Number of scenario samples \\
             $M_{\mathrm{eval}}$ & Number of evaluation samples \\
             $T$ & Number of trials
        \end{tabular}
        
        \For{$t = 1, 2, \dots, T$}
            \State Draw $\bm{\tripDist}_{\mathrm{sc}}^{(j)} \sim
            \Pbb$ for
            $j = 1, \dots, M_{\mathrm{sc}}$.
            \State Draw $\bm{\tripDist}_{\mathrm{eval}}^{(j)} \sim
            \Pbb$ for
            $j = 1, \dots, M_{\mathrm{eval}}$.
            \State $\bm{\batSizPriv}_t, \batSizPub_t, \hat{\alpha}_t \gets
            \texttt{BinarySearch}(
                \{ \bm{\tripDist}_{\mathrm{sc}}^{(j)} \}_{j=1}^{M_{\mathrm{sc}}},
                \{ \bm{\tripDist}_{\mathrm{eval}}^{(j)} \}_{j=1}^{M_{\mathrm{eval}}},
                \alpha
            )$
        \EndFor
        \State \Return $(\bm{\batSizPriv}_t, \batSizPub_t)$,
        where $t = \arg\min_{s \in [T]} \hat{\alpha}_{s}$.
    \end{algorithmic}
    \label{alg:conservatism-reduction}
\end{algorithm}

\begin{algorithm}[h]
    \caption{Binary search to reduce conservatism}
    \begin{algorithmic}[1]
        \State \textbf{Inputs}:
        \begin{tabular}[t]{l l}
            $\bm{\tripDist}_{\mathrm{sc}}^{(j)}, \; j = 1, \dots, M_{\mathrm{sc}}$ &
            Scenario samples \\
            $\bm{\tripDist}_{\mathrm{eval}}^{(j)}, \; j = 1, \dots, M_{\mathrm{eval}}$ &
            Evaluation samples \\
            $\alpha$ & Target reliability
        \end{tabular}
        \State Set $m_{\mathsf{lo}} = 1$, $m_{\mathsf{hi}} = m$.
        \While{$\abs{m_{\mathsf{hi}} - m_{\mathsf{lo}}} > 1$}
            \State Set $m_{\mathsf{mid}} := \ceil{\frac{m_{\mathsf{lo}} + m_{\mathsf{hi}}}{2}}$.
            \State Let ($\bm{\batSizPriv}, \batSizPub$)
            solve \eqref{eq:inner-scenario} using
            $\set{
                \bm{\tripDist}^{(i)}_{\mathrm{sc}}
                \mid i \in [m_{\mathsf{mid}}]}$.
            \State Compute the empirical reliability
            \[
                \hat{\alpha} :=
                \frac{1}{k} \sum_{j = 1}^k
                \bm{1}\set{
                    \sum_{i=1}^N ([\bm{\tripDist}_{\mathrm{eval}}^{(j)}]_i - \batSizPriv_i)_+
                    \leq \batSizPub
                }
            \]
            \If{$\hat{\alpha} > \alpha$}
                \State $m_{\mathsf{hi}} \gets m_{\mathsf{mid}}$
            \Else
                \State $m_{\mathsf{lo}} \gets m_{\mathsf{mid}}$.
            \EndIf
        \EndWhile
        \State \Return $\bm{\batSizPriv}$, $\batSizPub$, $\hat{\alpha}$.
    \end{algorithmic}
    \label{alg:binary-search}
\end{algorithm}

%% file: data_model.tex
To assess the potential of resource sharing to reduce deployed battery capacity in practice, we conduct an empirical study using real-world mobility data.
The data was collected as part of a study by the Puget Sound Research Council on the driving behavior of approximately 400 vehicles located in the Seattle metropolitan area between November 2004 and April 2006, and is publicly available through the National Renewable Energy Laboratory's Transportation Secure Data Center~\cite{tsdc2022}.
During the study, GPS data loggers were installed into each vehicle and collected information on the timing and distance of every trip taken by the vehicle's driver.
Figure~\ref{fig:daily-commute-distance-histograms} illustrates distributions over daily mileage for each driver and total daily mileage summed across all drivers.
Notice that daily travel distances are short relative to typical EV battery ranges: 85\% of daily trips are less than $50$ miles long, and on 54\% days, the total daily mileage across customers is less than ten thousand.

\begin{figure*}
\centering
\includegraphics[width=0.9\linewidth]{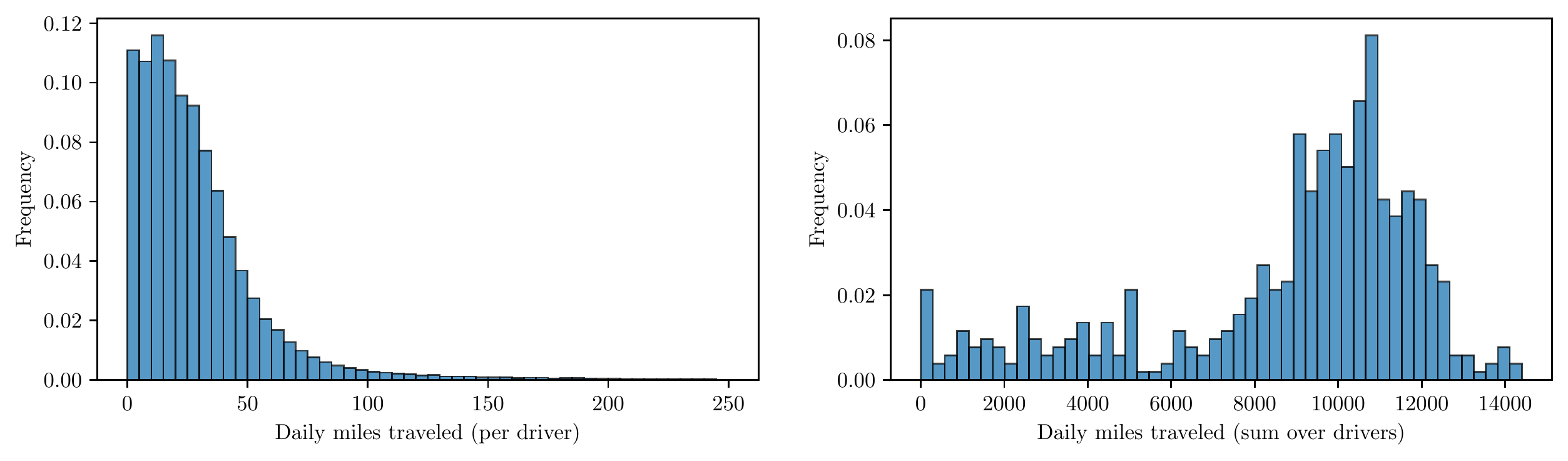}
\caption{Distributions over daily driving distance in the Puget Sound mobility dataset. \textbf{Left}: Distribution over daily mileage for each driver and day when travel occured. \textbf{Right}: Distribution over total daily mileage summed across drivers for each day of data collection.}
\label{fig:daily-commute-distance-histograms}
\end{figure*}

To simulate energy requirements from the daily mileage data, we assume that EVs have an energy efficiency of three miles per kWh \cite{inl2010gas}. 
Additionally, the distribution over each driver's daily energy consumption is modeled as a histogram with binwidth approximately two kWh. 
Driver energy requirements are then simulated by sampling from these distributions.

%% file: study_results.tex
In these empirical studies, we quantify the impacts of resource sharing by comparing the amounts of battery capacity required to achieve a given target reliability level in settings with and without a shared resource.
Throughout the studies, we assume that all drivers have the same target reliability level, i.e. that $\probSucc_i = \probSucc$ for all $i \in [N]$. 
For a given battery \textit{configuration} (i.e., a choice of personal battery capacities $\bm{\batSizPriv}$ and shared battery capacity $\batSizPub$) we estimate the reliability level associated with the configuration through its empirical reliability in each setting. 
That is, in the non-shared setting, the battery capacity associated with a given reliability level $\alpha$ is approximated by the sum of the drivers' empirical $\alpha$-quantiles.
In the shared setting, we use the inner approximation Problem~\eqref{eq:inner} and Algorithm~\ref{alg:conservatism-reduction} to find candidate battery configuration and evaluate the empirical reliability associated with the configuration under various allocation rules.
Specifically, we consider a given configuration's reliability under the proportional, first-come, first-serve (FCFS), and utilitarian allocation rules described below.
\paragraph{Proportional allocation} Under the proportional
allocation rule, every driver receives a fraction of the shared capacity
that is proportional to their \textit{shortfall} (the difference between their personal battery capacity and realized energy requirement):
\[
    f_{i}^{\text{prop}}(\bm{\batSizPriv}, \bm{\tripDist}, \batSizPub)
    := \batSizPub \cdot
    \frac{(\tripDist_i - \batSizPriv_i)_+}{\sum_{j=1}^N
        (\tripDist_j - \batSizPriv_j)_+
    }.
\]
\paragraph{First-come first-served (FCFS)} The FCFS rule
assumes that drivers request a portion of the shared resource according to a certain order,
given as a permutation $\pi: [N] \to [N]$:
\[
    f_{i}^{\text{FCFS}}(\bm{\batSizPriv}, \bm{\tripDist},
    \batSizPub) :=
    \begin{cases}
        (\tripDist_i - \batSizPriv_i)_+, & \text{if
        $
        \sum_{j=1}^{\pi(i)} (\tripDist_j - \batSizPriv_j)_+
        \leq \batSizPub
        $}
        , \\
        0, & \text{otherwise}
    \end{cases}.
\]
For a fixed sample of realized driver energy requirements, we simulate the FCFS allocation rule by drawing a random permutation of $[N]$ and allocating available shared capacity to drivers in that order.
\paragraph{Utilitarian} Under ``utilitarian" rules, the objective of the decision maker distributing the shared resource is to maximize the number of drivers whose energy requirements are met through the shared resource.  Under this allocation rule, the drivers are sorted in increasing order of shortfall, and resources are disbursed according to this ordering. That is, the utilitarian allocation rule disburses resources as follows:
\begin{align*}
    f_{i}^{\text{util}}(\bm{\batSizPriv}, \bm{\tripDist},
    \batSizPub) &=
    \begin{cases}
        (\tripDist_i - \batSizPriv_i)_+, & \text{if
        $
        \sum_{j=1}^{\pi^{\star}(i)} (\tripDist_j - \batSizPriv_j)_+
        \leq \batSizPub
        $}
        , \\
        0, & \text{otherwise}
    \end{cases}, \\
    \text{where} \;\; \pi^{\star}(i) &\leq \pi^{\star}(j) \Leftrightarrow
    (\tripDist_i - \batSizPriv_i)_+ \leq
    (\tripDist_j - \batSizPriv_j)_+.
\end{align*}

For a particular allocation rule and battery configuration, the associated empirical reliability is calculated as the minimum over customer-level empirical reliabilities (i.e., the largest value $\hat{\probSucc}$ such that all customers have empirical energy requirement satisfaction probability at least $\hat{\probSucc}$). 
This calculation is described in detail in Appendix~\ref{sec:app-rel}

\subsubsection{Effect of the target reliability level}
First, we demonstrate how battery capacity size requirements scale with target reliability level $\probSucc$ in the shared and non-shared settings.
Figure~\ref{fig:reliability-over-capacity} illustrates the empirical reliability level as a function of the average (per-driver) deployed battery capacity for systems with 25, 50, and 100 drivers and target reliabilities $\alpha \in \{0.5, 0.505, \dots, 0.995\}$.
In each sub-figure, the purple line labeled ``Non-shared'' depicts the capacity-reliability frontier in the setting without resource sharing and three scatter plots of differing colors show the reliability levels associated with various candidate battery configurations for each of the three allocation rule considered.
Additionally, the `efficient frontier' of each allocation rule (i.e., the largest reliability level associated with a particular capacity level) is depicted by solid lines.

For target reliability levels higher than $0.70$, the amount of battery capacity required to achieve a particular reliability is lower in the shared setting than in the non-shared settings for each system size and allocation rule considered.
Furthermore, the benefits of sharing (as measured by a reduction in battery capacity requirements) increase as the target reliability level increases. 
For example, as depicted in Figure~\ref{fig:reliability-over-capacity}, for $N=25$ the difference between the capacity required to achieve a reliability level of 0.75 with and without sharing ranges between about 5-10 kWh per driver (depending on the allocation rule used). 
By contrast, for a target reliability level of 0.85 the difference is larger, between about 10-20 kWh per driver. 

It is worth noting that, as evidenced by the scatter plots in Figure \ref{fig:reliability-over-capacity}, the empirical level of reliability associated with each allocation rule can vary significantly. 
While the FCFS and utilitarian allocation rules are associated with similar reliability levels given a candidate configuration, the proportional allocation rule tends to be more conservative in the sense that the reliability level associated with a given configuration is lower than that of the other two rules.
In fact, due to its conservatism, the proportional allocation rule is associated with larger capacity requirements than the non-shared setting for low reliability levels (e.g., below 0.70 in the $N=25$ setting or below 0.67 in the $N=100$ setting).  

\begin{figure*}[h!]
\centering
\includegraphics[width=\linewidth]{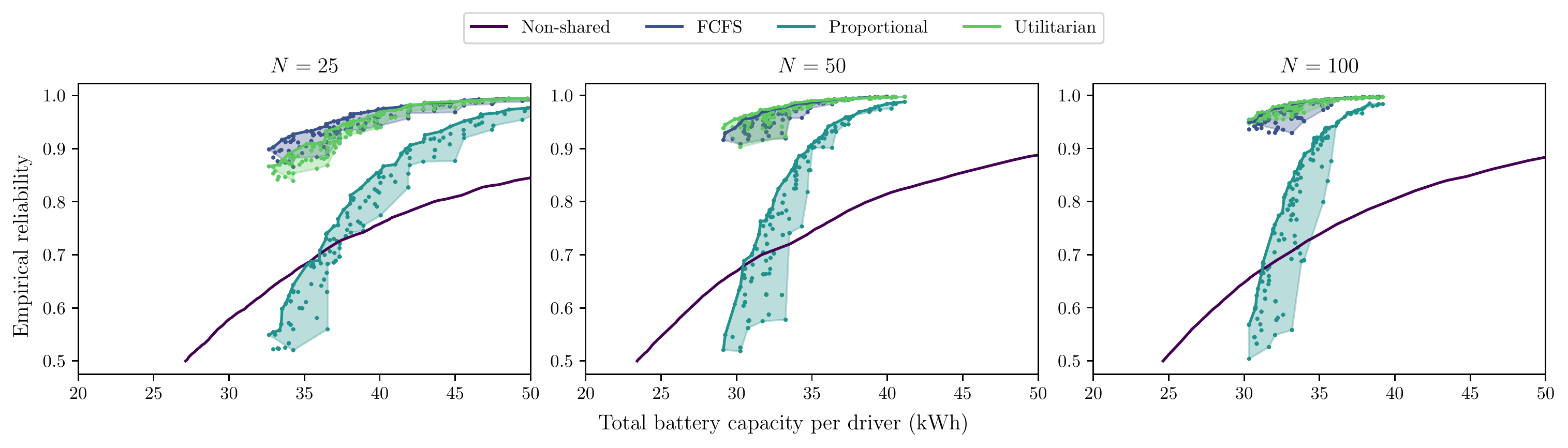}
\caption{Empirical reliability over deployed
battery capacity (averaged across drivers) for various allocation rules and numbers of drivers $N$. In each sub-figure, the capacity-reliability frontiers associated with the non-shared setting and each allocation rule considered in the shared setting are depicted by solid lines of varying colors. Additionally, candidate solutions considered in the shared settings and their associated reliabilities are plotted as colored scatter plots for each allocation rules considered.}
\label{fig:reliability-over-capacity}
\end{figure*}

\subsubsection{Effect of the number of drivers}
The benefits of resource sharing also increase with the number of drivers in the system.
Notice from Figure~\ref{fig:reliability-over-capacity} that as the number of customers increases from $N=25$ to $N=100$, the reliability associated with a given allocation rule and level of capacity increases.
For example, the reliability associated with 35 kWh of battery capacity per driver under a proportional allocation rule is approximately 0.67 for $N=25$, 0.89 for $N=50$, and 0.91 for $N = 100$.

To illustrate the benefits of resource sharing as a function of driver number in greater detail, Figure~\ref{fig:pct-reduction-by-clients} shows the relative reduction in deployed battery capacity as a function of $\numCust$ for three different reliability levels $\alpha \in \set{0.75, 0.85, 0.95}$.
More precisely, for each of the three target reliabilities considered, we vary $N \in \set{5, 25, \dots, 185}$ and plot the relative reduction in battery capacity,
\[
    %\mathrm{reduction} = 
    1 - \frac{\sum_{i = 1}^N \batSizPriv^{\mathsf{shared}}_i + \batSizPub^{\mathsf{shared}}}{\sum_{i = 1}^N \batSizPriv^{\mathsf{nonshared}}_i},
\]
where $\bm{\batSizPriv}^{\mathsf{shared}}, \batSizPub^{\mathsf{shared}}$ are the shared battery configurations found by Algorithm~\ref{alg:conservatism-reduction}
and $\bm{\batSizPriv}^{\mathsf{nonshared}}$ is the non-shared configuration determined by the empirical estimate of~\eqref{eq:ns-quants}. 
For each $\numCust$ and target $\probSucc$, we the conduct 20 independent trials and plot the median, interquartile, and interdecile ranges of the relative reduction in total battery capacity.
We find that the percentage reduction in capacity grows with the number of drivers $N$.
For example, for $\alpha = 0.85$, a system size of $N=5$ is associated with a median capacity reduction of 7\% while a system size of $N=185$ is associated with a median reduction of over 20\%.
Moreover, for sufficiently large $N$ and target reliability levels, the benefits of sharing can be large: for a target reliability of $\alpha = 0.95$, the availability of a shared resource enables at least a $50\%$ reduction in deployed battery capacity relative to the non-shared setting.

\begin{figure}[h]
\centering
\includegraphics[width=0.95\linewidth]{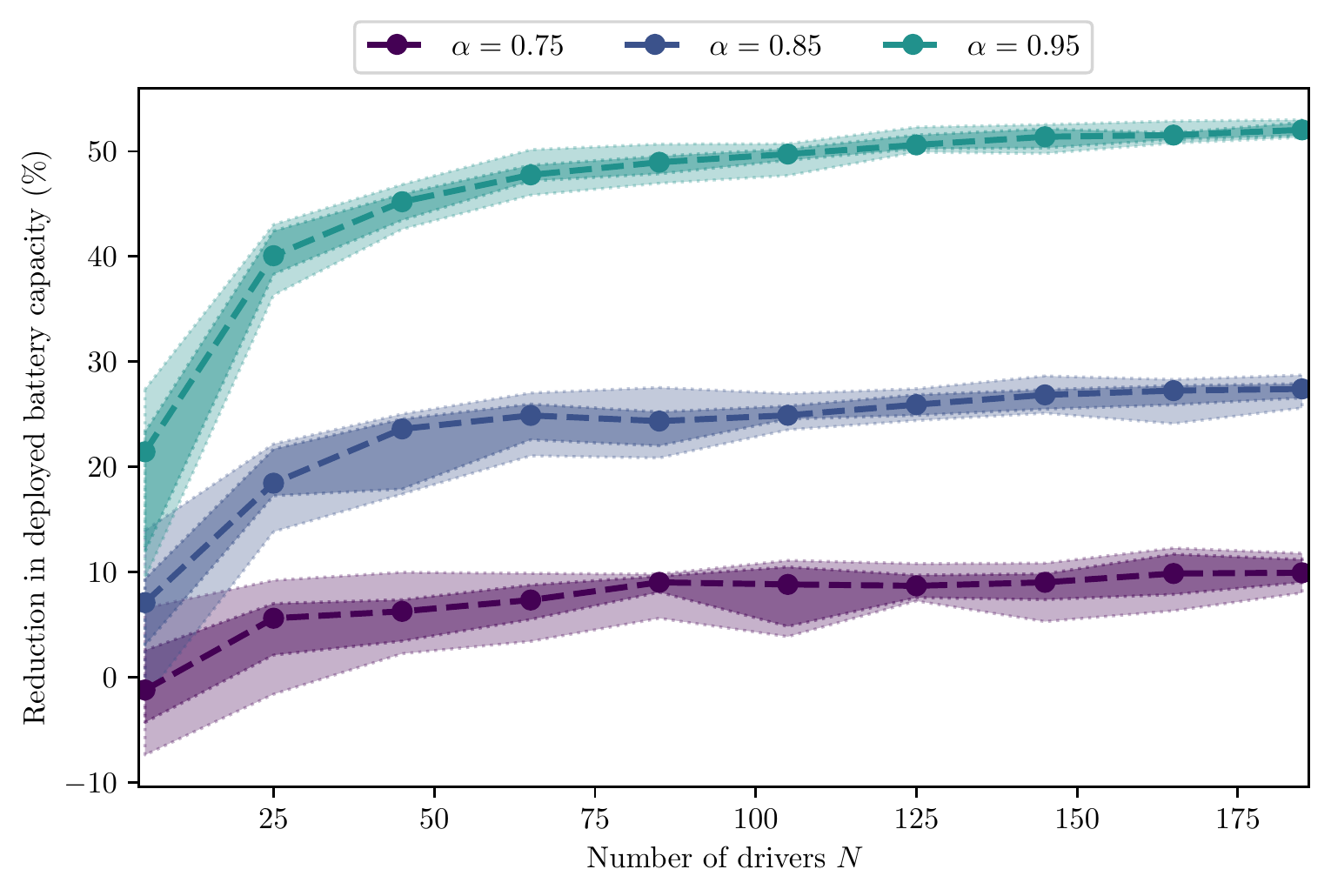}
\caption{Percentage reductions in total deployed battery capacity associated with resrouce sharing as a function
of the number of drivers in the system $N$ for three different reliability levels $\alpha$. For each system size $N$ and target reliability level $\probSucc$ considered, the median percent reduction (dashed line), its interquartile range (darkly shaded region), and its interdecile range (lightly shaded region) are shown.}
\label{fig:pct-reduction-by-clients}
\end{figure}

%% file: conclusion.tex
We consider a setting in which electric vehicle drivers are given access to a shared battery resource which can be used to complement their personal vehicle batteries.
From the perspective of a central decision maker, we formulate the problem of choosing the personal and shared battery capacities in order to minimize total deployed capacity while ensuring that driver mobility needs are met with high probability.
The resultant capacity planning problem is a chance constrained optimization problem and can be challenging to solve directly or even approximate using sampling-based methods. 
We derive a tractable inner approximation to the original capacity planning problem which is amenable to approximation through scenario methods.
To assess the potential of resource sharing to reduce total deployed battery capacity (relative to a setting without sharing) we conduct an empirical study using longitudinal mobility data from drivers in Puget Sound, Washington. 
The empirical results demonstrate that resource sharing has the potential to greatly reduce the amount of battery capacity deployed, and that benefits from sharing increase with number of drivers in the system and the target reliability level desired.
In particular, when driver target reliabilities are high (e.g., greater than 0.95), resource sharing can reduce the amount of deployed battery capacity by at least 40\%, even for moderately sized systems of 25 drivers or more. 

Our results suggest resource sharing has significant potential to reduce deployed battery capacity, and merits further exploration.
There are several interesting directions for future work.
First, the stylized model considered in this paper can be refined to reflect various practical considerations. 
For example, although we assume that the shared resource can be divided into units of arbitrary size, commercially available range extension resources are typically of fixed size, giving rise to a problem with integer constraints. 
Additionally, while we assume the perspective of a central decision maker who has control over the size of both the personal and shared batter capacity, it may be more reasonable to model each driver and the central planner as individual agents.
Such an assumption would give rise to a game-theoretic model, in which one might attempt to characterize an allocation rule that induces socially optimal behavior from a mechanism design or cooperative game theory perspective.

%% file: appendix.tex
\subsection{Nonconvexity of proportional allocation rule} \label{sec:app-nonconvex}
We argue that the proportional allocation rule leads to bilinear
constraints in the scenario approximation of~\eqref{eq:opt-shared}.
Consider the case where $N = 2$ and relabel $v_i := \tripDist_i - \batSizPriv_i$.
The constraints read
\[
\begin{aligned}
    -v_1 + \frac{\batSizPub (v_1)_+}{(v_1)_+ + (v_2)_+} &\geq 0, \\
    -v_2 + \frac{\batSizPub (v_2)_+}{(v_1)_+ + (v_2)_+} &\geq 0.
\end{aligned}
\]
On $\set{(v_1, v_2, \batSizPub) \mid (v_1)_+ + (v_2)_+ > 0}$, these are
equivalent to
\[
    -v_i ((v_1)_+ + (v_2)_+) + \batSizPub (v_i)_+ \geq 0.
\]
In the above, $\batSizPub$ is coupled with
$\max(\tripDist_i - \batSizPriv_i, 0)$, which does not lead to a
convex constraint set in general.

\subsection{Example: Intuition behind shared allocation model} \label{sec:app-gaus_ex}
To illustrate why we expect~\eqref{eq:inner} to yield better solutions than the non-shared setting, we
compare the objective value attained by~\eqref{eq:opt-non-shared} in the non-shared
setting with an ``obvious'' choice of $\set{\batSizPriv_i}$ and $\batSizPub$ under the following assumption:
\begin{center}
    \textbf{Assumption.} The trip energies satisfy
    $\bm{\tripDist} \sim \cN(
        \mu\bm{1},
        \sigma^2 I_N
    )$.
\end{center}
Fix $\alpha_i \equiv \alpha$ such that \( F_i^{-1} ( \alpha_i) = \mu(1 + \sigma) \).
Then the optimal value of Problem~\eqref{eq:opt-non-shared} is
\[
    \mathsf{opt}^{*}_{\mathrm{ns}} = \sum_{i = 1}^N F_i^{-1} (\alpha_i) = N\mu(1 + \sigma).
\]
Now, we demonstrate a feasible solution for Problem~\eqref{eq:inner} attaining
a strictly better objective value. In particular, fix 
\[
    \batSizPriv_i = \expec{\tripDist_i} = \mu, \quad \text{for all $i = 1, \dots, N$}.
\]
Under this choice, the sum of shortfalls is distributed as
\[
    \sum_{i = 1}^N \max\set{0, \tripDist_i - \batSizPriv_i} \overset{(d)}{=}
    \sum_{i = 1}^N \max\set{0, \zeta_i}, \quad \zeta_i \sim \cN(0, \sigma^2)
\]
Letting $Z_i := \max\set{0, \zeta_i}$, each term has expected value
\[
    \expec{Z_i} = \expec{g_i \bm{1}\set{g_i \geq 0}} =
    \frac{1}{2} \expec{\abs{g_i}} = \frac{\sigma}{2} \sqrt{\frac{2}{\pi}} =
    \frac{\sigma}{\sqrt{2 \pi}}.
\]
Therefore, we can rewrite the chance constraint as
\begin{equation}
    \prob{S_{N} \leq \batSizPub} \geq \alpha, \quad S_{N} := \sum_{i = 1}^N Z_i.
\end{equation}
Every $Z_i$ is a $\sigma$-subgaussian random variable, so by standard concentration inequalities (see, e.g.~\cite[Chapter 2]{Ver18}):
\begin{align}
    \prob{S_{N} - \frac{N \sigma}{\sqrt{2\pi}} \geq t}
    \leq \exp\left(-\frac{c t^2}{\sigma^2 N}\right)
\end{align}
Choosing $t \asymp \sigma \sqrt{N} \log\left(\frac{1}{1 - \alpha}\right)$ satisfies the
desired inequality. Therefore, the optimal value satisfies
\[
    \mathsf{opt}_{\mathrm{s}}^* \leq \batSizPub + \sum_{i = 1}^N \batSizPriv_i
    = \frac{N \sigma}{\sqrt{2\pi}} + \sigma\sqrt{N} \log\left(\frac{1}{1 - \alpha}\right)
    + N \mu,
\]
and the difference between objective values is
\begin{align*}
    \mathsf{opt}_{\mathrm{ns}}^* - \mathsf{opt}_{\mathrm{s}}^* &\geq
    N\sigma \left(1 - \sqrt{\frac{1}{2\pi}}
    - \frac{\sigma \log(\frac{1}{1-\alpha})}{\sqrt{N}} \right) \\
    &=
    \Omega(N \sigma).
\end{align*}
In particular, more clients $N$ in the system imply bigger
deployed capacity savings by the shared design.

\subsection{Estimating empirical reliability by sampling} \label{sec:app-rel}
Suppose that a configuration $\left(\bm{\batSizPriv}, \batSizPub, \set{g_{i}}_{i=1}^N\right)$ is
provided, where $g_i$ is the constraint with respect to which reliability should be estimated. To estimate the reliability level achieved by that configuration, we use the following procedure:
\begin{enumerate}
    \item Generate $m$ independent realizations $\set{\bm{\tripDist}^{(j)}}_{j=1}^m$.
    \item Compute the client-level average satisfaction
    \[
        S_i^m = \frac{1}{m}
        \sum_{j=1}^m \bm{1}\set{
            g_{i}(
                \bm{\batSizPriv}, \bm{\tripDist}, \batSizPub
            ) \geq 0
        }.
    \]
    \item Output the minimum satisfaction over clients as the reliability level:
    \(
        \widehat{\alpha} := \min_{i = 1}^N S_{i}^m.
    \)
\end{enumerate}
To calculate the necessary number of samples $m$, we note
\begin{equation*}
\begin{aligned}
    X_i^{(j)} := \bm{1}\set{
        g_{i}(\bm{\batSizPriv},
              \bm{\tripDist},
              \batSizPub)}, \;\;
    \expec{X_i^{(j)}} = \Pr(
        g_i(\bm{\batSizPriv}, \bm{\tripDist}, \batSizPub)
        \geq 0
    ).
\end{aligned}
\end{equation*}
Consequently, $X_i^{(1)}, \dots, X_i^{(m)}$ are independent and identically distributed
Bernoulli random variables for fixed $i \in [N]$. A Chernoff bound followed by a union
bound yields
\begin{equation}
    \prob{\max_{i \in [N]}\abs{
        S_{i}^m - \expec{S_i^m}
    } \geq \varepsilon}
    \leq 2N\exp\left(
        -\frac{m \varepsilon^2}{4}
    \right),
    \label{eq:chernoff-samples}
\end{equation}
For a target error $\varepsilon$ and confidence level $\delta$ we require
\[
    m \geq \frac{4 \log(2N / \delta)}{\varepsilon^2} \quad
    \text{samples}.
\]